\documentclass[10pt,letterpaper]{article}
\usepackage[top=0.85in,left=2.75in,footskip=0.75in]{geometry}

\usepackage{amsmath,amssymb}

\usepackage{changepage}

\usepackage[utf8x]{inputenc}

\usepackage{textcomp,marvosym}

\usepackage{cite}

\usepackage{nameref,hyperref}

\usepackage[right]{lineno}

\usepackage{microtype}
\DisableLigatures[f]{encoding = *, family = * }

\usepackage[table]{xcolor}

\usepackage{array}

\usepackage{amsmath, amsfonts, amssymb, amsthm, enumitem, xcolor, courier, multicol, caption,subcaption, algpseudocode, algorithm, graphicx, float, tkz-graph, tkz-berge}

\newtheorem{proposition}{Proposition}
\newtheorem{corollary}{Corollary}[proposition]

\newcolumntype{+}{!{\vrule width 2pt}}

\newlength\savedwidth

\raggedright
\usepackage[aboveskip=1pt,labelfont=bf,labelsep=period,justification=raggedright,singlelinecheck=off]{caption}

\makeatletter
\renewcommand{\@biblabel}[1]{\quad#1.}
\makeatother

\date{}

\usepackage{lastpage,fancyhdr,graphicx}
\usepackage{epstopdf}
\fancyheadoffset[L]{2.25in}
\fancyfootoffset[L]{2.25in}
\begin{document}
\vspace*{0.2in}

\begin{flushleft}
{\Large
\textbf\newline{Graphettes: Constant-time determination of graphlet and orbit identity including (possibly disconnected) graphlets up to size 8} % Please use "sentence case" for title and headings (capitalize only the first word in a title (or heading), the first word in a subtitle (or subheading), and any proper nouns).
}
\newline
\\
Adib Hasan\textsuperscript{1},
Po-Chien Chung\textsuperscript{2},
Wayne Hayes\textsuperscript{2*},
\\
\bigskip
\textbf{1} Ananda Mohan College, Mymensingh, Bangladesh
\\
\textbf{2} Dept. of Computer Science, University of California, Irvine, California, USA
\\
\bigskip

% Insert additional author notes using the symbols described below. Insert symbol callouts after author names as necessary.
% 
% Remove or comment out the author notes below if they aren't used.
%
% Primary Equal Contribution Note
%\Yinyang These authors contributed equally to this work.

% Additional Equal Contribution Note
% Also use this double-dagger symbol for special authorship notes, such as senior authorship.
%\ddag These authors also contributed equally to this work.

%\textcurrency Current Address: Dept/Program/Center, Institution Name, City, State, Country % change symbol to "\textcurrency a" if more than one current address note
% \textcurrency b Insert second current address 
% \textcurrency c Insert third current address


% Group/Consortium Author Note
%\textpilcrow Membership list can be found in the Acknowledgments section.

% Use the asterisk to denote corresponding authorship and provide email address in note below.
* Corresponding author: whayes@uci.edu

\end{flushleft}
\section*{Abstract}
{\it Graphlets} are small connected induced subgraphs of a larger graph $G$. Graphlets are now commonly used to quantify local and global topology of networks in the field. Methods exist to exhaustively enumerate all graphlets (and their orbits) in large networks as efficiently as possible using {\it orbit counting equations}. However, the number of graphlets in $G$ is exponential in both the number of nodes and edges in $G$. Enumerating them all is already unacceptably expensive on existing large networks, and the problem will only get worse as networks continue to grow in size and density. Here we introduce an efficient method designed to aid statistical {\em sampling} of graphlets up to size $k=8$ from a large network.  We define {\em graphettes} as the generalization of graphlets allowing for disconnected graphlets. Given a particular (undirected) graphette $g$, we introduce the idea of the {\em canonical} graphette $\mathcal K(g)$ as a representative member of the isomorphism group $Iso(g)$ of $g$. We compute the mapping $\mathcal K$, in the form of a lookup table, from all $2^{k(k-1)/2}$ undirected graphettes $g$ of size $k\le 8$ to their canonical representatives $\mathcal K(g)$, as well as the permutation that transforms $g$ to $\mathcal K(g)$. We also compute all automorphism orbits for each canonical graphette. Thus, given any $k\le 8$ nodes in a graph $G$, we can in constant time infer which graphette it is, as well as which orbit each of the $k$ nodes belongs to. Sampling a large number $N$ of such $k$-sets of nodes provides an approximation of both the distribution of graphlets and orbits across $G$, and the orbit degree vector at each node.

% Please keep the Author Summary between 150 and 200 words
% Use first person. PLOS ONE authors please skip this step. 
% Author Summary not valid for PLOS ONE submissions.   
\section*{Author summary}
{\it Graphlets} are small subgraphs of a larger network. They have been used extensively for over a decade in the analysis of social, biological, and other networks. Unfortunately it is extremely expensive to exhaustively enumerate all graphlets appearing in a large graph, requiring days or weeks of computer time for recent large networks. Here we introduce a novel method for statistically sampling graphlets from large graphs. The time required does not depend upon the size of the input network, but instead upon the number of samples desired.  In addition, existing methods only look at graphlets up to size 5 or 6; we allow graphlets up to size 8, which significantly improves on the sensitivity and specificity of network analysis. Our method will allow graphlets to be efficiently utilized to analyze networks of arbitrary size going into the future.
%\linenumbers

% Use "Eq" instead of "Equation" for equation citations.
\section*{Introduction}
Network comparison is a growing area of research. In general the problem of complete comparison of large networks is intractable, being an $NP$-complete problem \cite{Cook:1971:CTP:800157.805047}. Thus, approximate heuristics are needed. Networks have been compared for statistical similarity from a high-level using simple, easy-to-calculate measures such as the degree distribution, clustering co-efficients, network centrality, among many others  \cite{NewmanNetworks2010,emmert2016fifty}. While more sophisticated methods such as spectral analysis \cite{WilsonZhu2008Spectral,ThorneStrumpf2012} and topological indices \cite{dehmer2014interrelations} have been useful, the study of small subnetworks such as {\it motifs} \cite{Milo2002Motifs} and {\it graphlets} \cite{Przulj2004Graphlets,PrzuljGDD2007} have become popular. They have been used extensively to globally classify highly disparate types of networks \cite{Przulj2014HiddenLanguage} as well as to aid in local measures used to {\it align} networks \cite{GRAAL,LGRAAL,MAGNA,MamanoHayesSANA}. 
\begin{figure*}
\centering
\includegraphics[width=0.75\linewidth,angle=-90]{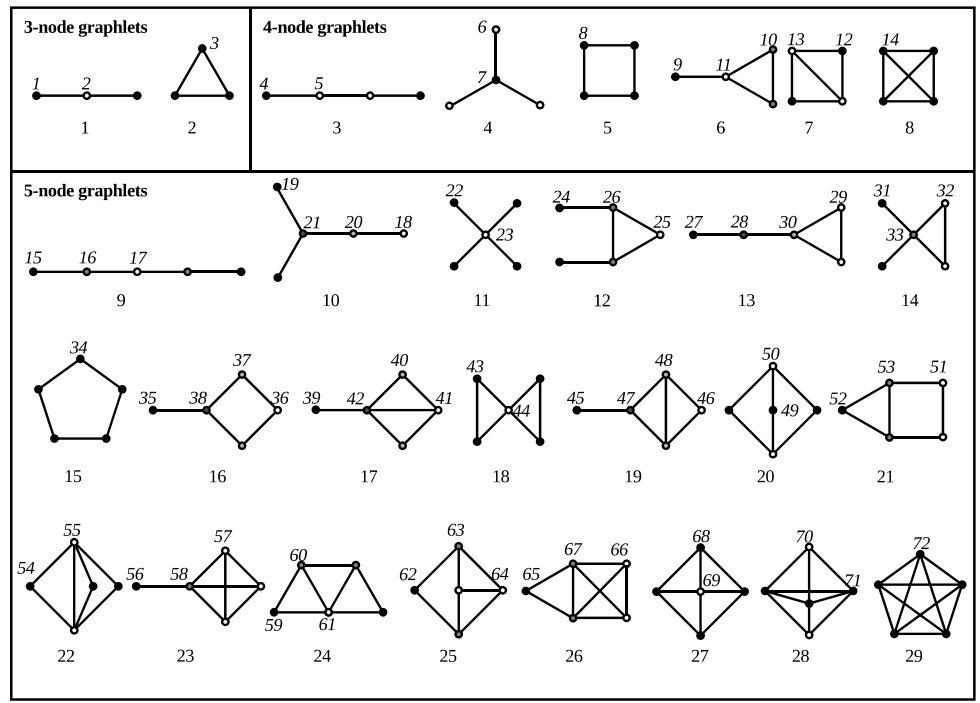}
\caption{All (connected) graphlets of sizes $k=3,4,5$ nodes, and their automorphism orbits; within each graphlet, nodes of equal shading are in the same orbit. The numbering of these graphlets and orbits were created by hand \cite{Przulj2004Graphlets} and do not correspond to the automatically generated numbering used in this paper. The figure is taken verbatim from  \cite{Melckenbeeck1EtAl2016}.}
\label{fig:graphlets5}
\end{figure*}

A {\em graphlet} is a small, connected, induced subgraph $g$ of a larger graph $G$. Given a particular graphlet $g$, the {\it automorphism orbits} of $g$ are the sets of nodes that are topologically identical to each other inside $g$.
Graphlets and their automorphism orbits with up to $k=5$ nodes were first introduced in 2004 \cite{Przulj2004Graphlets}, and are depicted in Fig \ref{fig:graphlets5}. Recently, automated methods have been created that can enumerate, in a larger graph, all graphlets and their automorphism orbits up to graphlet size $k=5$  \cite{ORCA} and subsequently to any $k$   \cite{Melckenbeeck1EtAl2016}, although the latter authors only applied it up to $k=6$. Unfortunately, we have found that these methods take a very long time (hours to days) even just to count graphlets up to size $k=5$ on some large biological networks, such as those in BioGRID  \cite{Chatr-aryamontri01012013}.  It is not clear that such methods, especially for even larger $k$, will be applicable to the coming age of ever bigger networks, since the total number of graphlets appearing in a large network tends to increase exponentially with both $k$ (the graphlet size) and $n$ (the number of nodes in the large network).  Eventually, an exhaustive enumeration of all graphlets appearing in a large network may become infeasible simply due to the number of graphlets that need to be enumerated, even under the optimization of using orbit counting equations. On the other hand, graphlets are too useful to abandon as a method of quantifying the topological structure of graphs. An achievable alternative for a large network $G$ is to statistically sample its graphlets rather than exhaustively enumerate them.  Additionally, such sampling could be useful with the recent advent of comprehensive biological network {\it databases}   \cite{NDEx2017}: each sampled graphlet would act as a seed for local matching between larger networks, similar to how {\it k-mers} (short sequences of length $k$) are used for seed-and-extend sequence matching in BLAST  \cite{blast}.

To efficiently create a statistical sample of graphlets in a large network $G$, one must be able to take an arbitrary set of $k$ nodes from $G$, and efficiently (preferably in constant time) determine both {\em which} graphlet is represented, as well as the automorphism orbits of each of the $k$ nodes. Here, we solve this problem both by enumerating all graphlets (and their disconnected counterparts, which we term {\it graphettes}) and their automorphism orbits up to graphettes of size $k=8$. We present a method that creates a lookup table that can quickly determine the graphette identity of any $k$ nodes, as well as their automorphism orbits. Since the lookup table required significant time to pre-compute for $k=7$ (a few hours on a single core) and $k=8$ (hundreds of CPU weeks on a cluster), we provide the actual lookup tables for these values of $k$ online at \url{http://github.com/Neehan/Faye}.

\section*{Materials and methods}
\subsection*{Definitions and notations}

Given a graph $G$ on $n$ nodes, a $k$\textit{-graphette} is a (not necessarily connected) induced subgraph $g$ on any set of $k$ nodes of $G$.  There are many ways one could choose the $k$ nodes, for example (i) choosing $k$ nodes uniformly at random from $G$, or (ii) performing a local search around some node $u$. We expect the former to be useful only in dense networks, while the latter is probably more useful in sparse networks because most random sets of $k$ nodes in a sparse graph will be highly disconnected and thus not very informative. One could also (iii) perform edge-based selection (with local expansion) to ensure dense regions are sampled more frequently than sparse regions  \cite{rahman2014graft}; still other methods have been suggested  \cite{prvzulj2006efficient}.

Given a set of $k$ nodes, we wish to quickly ascertain which graphette is represented, and which automorphism orbits each of the $k$ nodes belong to. To do that we need a canonical list of graphettes and their orbits, and a fast way to determine which canonical graphette is represented by any permutation of $k$ nodes. Here we demonstrate how, if $k$ is fixed and relatively small ($k\le 8$ in our case), this can be accomplished in constant time by pre-computing and storing a lookup table indexed by a bit vector representation of the lower triangular matrix of the (undirected) adjacency matrix of the induced subgraph.  Given such an index, the value associated with that index identifies the canonical graphette (a canonical ordering of the nodes for that graphette).  We also pre-compute the automorphism orbits of all the canonical graphettes.  Thus, by reversing the lookup table we can, in constant time, infer the orbit identity of each of the $k$ nodes in that $k$-graphette. As a corrollary, we can also update the (statistically sampled) {\it graphette orbit degree vector} of each of the $k$ nodes, similar to the graphlet degree vector \cite{PrzuljGDD2007}.

We use the following abbreviations and notations throughout:

\noindent\begin{tabular}{r p{0.7\columnwidth}}
$G(V,E)$                    & The Graph with nodes $V$ and edges $E$\\
$\mathcal V(G)$             & The set of nodes of graph $G$\\
$\mathcal E(G, u, v)$          & The boolean value denoting connectivity between nodes $u$ and $v$ of graph $G$\\
$\Longleftrightarrow$, iff  & If and only if\\
$|S|$                       & The number of elements in set $S$.\\
$Adj(G)$                    & The adjacency matrix representation of graph $G$\\
$Aut(G)$                    & The set of automorphisms of graph $G$\\
$\mathcal K(g)$                & Canonical isomorph of graphette $g$\\
\end{tabular}
%The Petersen graphs
\begin{figure}[thb]
%	\centering
%	\GraphInit[vstyle=Normal] 
%	\SetVertexMath
%	\SetVertexNoLabel
%	\tikzset{VertexStyle/.style = {%
%			shape = circle,
%			inner sep = 1pt,
%			outer sep = 0pt,
%			minimum size = 13pt,
%			draw}}
%	\begin{tikzpicture}[scale=.35, rotate=18] 
%	\grPetersen[form=1,RA=5,RB=2.5]%
%	\AssignVertexLabel{a}{0,1,2,3,4}
%	\AssignVertexLabel{b}{5,6,7,8,9}
%	\end{tikzpicture}
%	\begin{tikzpicture}[scale=.35] 
%	\grPetersen[form=2,RA=5,RB=3]%
%	\AssignVertexLabel{a}{0,1,2,3,8,5}
%	\AssignVertexLabel{b}{7,6,4,9}
%	\end{tikzpicture}
%	\begin{tikzpicture}[scale=.35] 
%	\grPetersen[form=3,RA=5,RB=3]%
%	\AssignVertexLabel{a}{0,1,2,3,4,9,6,8,5,7}
%	\end{tikzpicture}
	\includegraphics[width=0.75\linewidth]{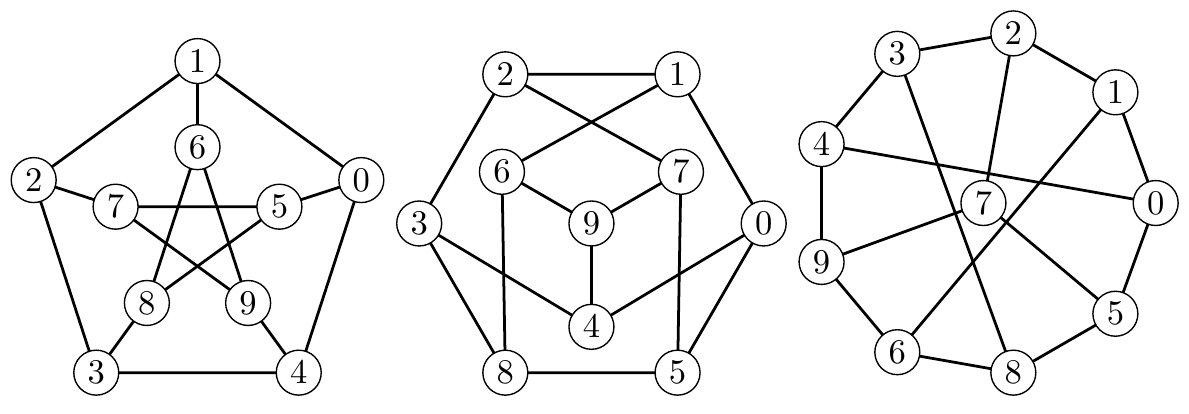}
	\caption{Three isomorphic representations of the Petersen graph.}
	\label{fig:Petersen}
\end{figure}
\subsection*{Canonization of graphettes}
If graphs $G$ and $H$ are isomorphic, it essentially means they are exactly the same graph, but drawn differently. For example, Fig \ref{fig:Petersen} shows three different drawings of the Petersen graph. Technically,
an isomorphism between networks $G$ and $H$ is a permutation $\pi:\mathcal V(G)\to\mathcal V(H)$ so that
\[\mathcal E(G,u,v)\Longleftrightarrow\mathcal E(H,\pi(u),\pi(v)),\]

\begin{figure}[tbh]
	\centering
%	\SetGraphUnit{0.8}
%	\GraphInit[vstyle=Classic]
%	\tikzset{VertexStyle/.style = {%
%			shape = circle,
%			fill =  black,
%			inner sep = 1pt,
%			outer sep = 1pt,
%			minimum size = 5pt,draw}}
%	\begin{subfigure}[b]{0.4\columnwidth}
%		\centering
%		\begin{tikzpicture}[rotate=90]
%		\Vertices[Math, Lpos=90]{circle}{w,x,y}
%		\end{tikzpicture}
%	\end{subfigure}
%	\begin{subfigure}[b]{0.4\columnwidth}
%		\centering
%		\begin{tikzpicture}[rotate=90]
%		\Vertices[Math, Lpos=90]{circle}{w,x,y}
%		\Edges(w,x)
%		\end{tikzpicture}
%	\end{subfigure}
%	
%	\begin{subfigure}[b]{0.4\columnwidth}
%		\centering
%		\begin{tikzpicture}[rotate=90]
%		\Vertices[Math, Lpos=90]{circle}{w,x,y}
%		\Edges(x,w,y)
%		\end{tikzpicture}
%	\end{subfigure}
%	\begin{subfigure}[b]{0.4\columnwidth}
%		\centering
%		\begin{tikzpicture}[rotate=90]
%		\GraphInit[vstyle=Classic]
%		\Vertices[Math, Lpos=90]{circle}{w,x,y}
%		\Edges(w,x,y,w)
%		\end{tikzpicture}
%	\end{subfigure}
	\includegraphics[width=0.75\linewidth]{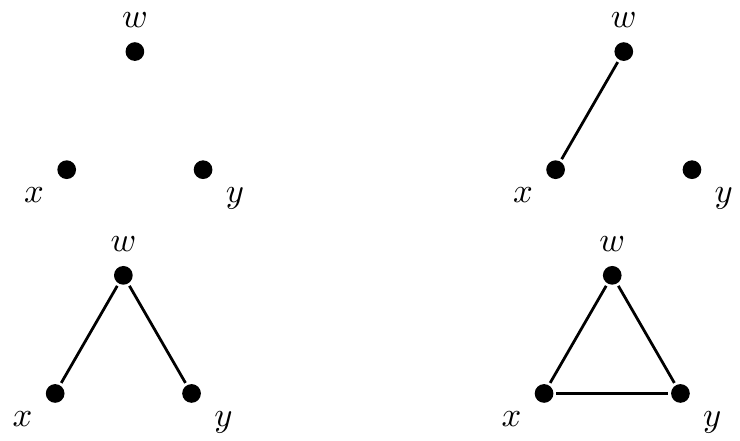}
	\captionsetup{width=0.6\textwidth}
	\caption{All the possible 3-graphettes.}
	\label{fig:3-graphettes}
\end{figure}

Consider a 3-graphette with nodes $w,x$ and $y$. There are only 4 possible such graphettes, depicted in Fig \ref{fig:3-graphettes}.
However, by permuting the order of the nodes, each of these graphettes can be represented by several isomorphic variants.  In order to determine if two graphettes are isomorphic, we will represent its (undirected) graph with the lower-triangle of its adjacency matrix.  We will place this lower-triangular matrix into a bit vector, resulting in a representation similar to existing ones for orbit identification  \cite{Melckenbeeck1EtAl2016}.

\begin{figure}
    \centering
    %\SetGraphUnit{0.8}
	%\GraphInit[vstyle=Classic]
	%\tikzset{VertexStyle/.style = {%
	%		shape = circle,
	%		fill =  black,
	%		inner sep = 1pt,
	%		outer sep = 1pt,
	%		minimum size = 5pt,draw}}
    %\begin{tabular}{rccc}
    %    & \begin{tikzpicture}[rotate=90]
	%	\Vertices[Math, Lpos=90]{circle}{w,x,y}
	%	\Edges(w,x)
	%	\end{tikzpicture} & 
	%	\begin{tikzpicture}[rotate=90]
	%	\Vertices[Math, Lpos=90]{circle}{w,x,y}
	%	\Edges(x, y)
	%	\end{tikzpicture} &
	%	\begin{tikzpicture}[rotate=90]
	%	\Vertices[Math, Lpos=90]{circle}{w,x,y}
	%	\Edges(w, y)
	%	\end{tikzpicture}\\
    %    &&&\\
    %    Matrix:  & \begin{tabular}{cccc}
    %              &$w$             & $x$     %      & $y$\\
    %              $x$&0               &               &  \\
    %              $y$&\bf 1   & 0             &  \\
    %              $z$&\bf 0   & \bf 0 & 0\\
    %              \end{tabular}
    %            & \begin{tabular}{cccc}
    %              &$w$           & $x$           & $y$\\
    %              $x$&0             &               &    \\
    %              $y$&\bf 0 & 0             & %   \\
    %              $z$&\bf 0 & \bf 1 & 0  \\
    %              \end{tabular}
    %            & \begin{tabular}{cccc}
    %              &$w$          & $x$        %    & $y$\\
    %              $x$&0            &         %       &  \\
    %              $y$&\bf 0 & 0             & % \\
    %              $z$&\bf 1 & \bf 0 & 0\\
    %              \end{tabular}\\
    %              &&&\\
    %    Bit vector: & \bf100 & \bf001 & %\bf010\\
    %\end{tabular}
    \includegraphics[width=0.8\linewidth, clip]{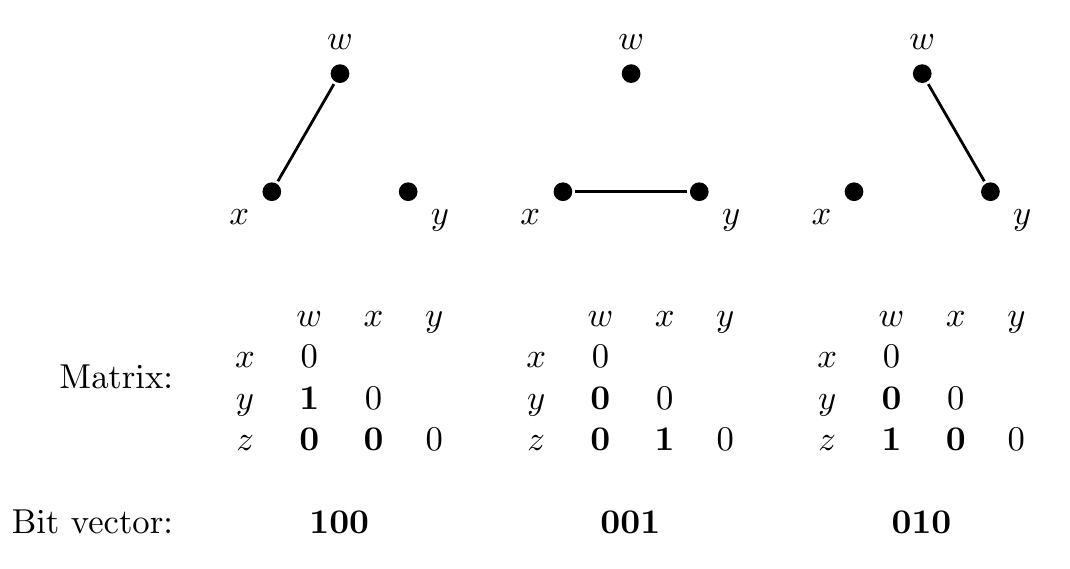}
    \caption{All 3-graphettes with exactly one edge; the {\it canonical} one is the one with lowest integer representation (the middle one in this case). Each of them is placed in a lookup table indexed by the bit vector representation of its adjacency matrix, pointing at the canonical one. In this way we can determine that it is the one-edge 3-graphette in constant time.}
    \label{fig:lookup-table}
\end{figure}

We now describe the idea of a {\it canonical representative} of each isomorph.
To provide an explicit example, consider Fig \ref{fig:lookup-table}, depicting the three isomorphic configurations of the 3-graphette that has exactly one edge.
In order to determine that these graphettes are all isomorphic, we take the bit vector representation depicted, and define the lowest-numbered bitvector among all the isomorphs as the {\it canonical} representative. All the other isomorphs in the lookup table point to it. 
In this way, every graph on 3 nodes can be efficiently mapped to its canonical 3-isomorph. 

\begin{table*}
%\begin{adjustwidth}{-2.25in}{0in}
\centering
\begin{tabular}{| l | l | l | l | l | l |}
\hline
$k$ & bits   & \#Graphs  & Space  & \#Canonicals & \#Orbits  \\
    & $b(k)$ & $2^{b(k)}$            & $b(k) 2^{b(k)}$     & $NC(k)$  & \\
\hline
%0   & 0          & 1      & 0       & 1             & 1\\
%\hline
1   & 0          & 1      & 0       & 1             & 1\\
\hline
2   & 1          & 2      & 0.25 B  & 2             & 2\\
\hline
3   & 3          & 8      & 3 B     & 4             & 6\\
\hline
4   & 6          & 64     & 48 B    & 11            & 20\\
\hline
5   & 10         & 1 K    & 1.25 KB & 34            & 90\\
\hline
6   & 15         & 32 K   & 60 KB   & 156           & 544\\
\hline
7   & 21         & 2 M    & 5.25 MB & 1044          & 5096\\
\hline
{\bf 8}   & {\bf 28}         & {\bf 256 M}  & {\bf 896 MB}  & {\bf 12346}         & {\bf 79264}\\
\hline
9   & 36         & 64 G   & 288 GB  & 274668        & 2208612\\
\hline
10  & 45         & 32 T   & 180 TB  & 12005168      & 113743760\\
\hline
11  & 55         &32 P    & 220 PB  & 1018997864    & 10926227136\\
\hline
12  & 66         &64 E    & 528 EB  & 165091172592  & 1956363435360\\
\hline
\end{tabular}
\captionsetup{width=1\textwidth}
\caption{For each value of $k$: the number of bits $b(k)=\frac{k(k-1)}{2}$ required to store the lower-triangle of the adjacency matrix for an undirected $k$-graphette; the number of such $k$-graphettes counting all isomorphs which is just $2^{b(k)}$; the number of canonical $k$-graphettes (this will be the number of unique entries in the above lookup table \cite{numcanon}, and up to $k=8$, 14 bits is sufficient); and the total number of unique automorphism orbits (up to $k=8$, 17 bits is sufficient)\cite{numorbit}. Note that up to $k=8$, together the lookup table for canonical graphettes and their canonical orbits fits into 31 bits, allowing storage as a single 4-byte integer, with 1 bit to store whether the graphette is connected (i.e., also a graphlet).
The suffixes K, M, G, T, P, and E represent exactly $2^{10}, 2^{20}, 2^{30}, 2^{40}, 2^{50}$ and $2^{60}$, respectively.}
\label{tab:k-graphette-canonical-orbit}
%\end{adjustwidth}
\end{table*}

We also automatically determine the number of {\it automorphism orbits} (see below) for each canonical isomorph. \autoref{tab:k-graphette-canonical-orbit} represents, for various values of $k$, the number of bits $b(k)$ required to store the lower-triangular matrix of all graphettes on $k$ nodes (i.e., the length of the bit vector used to store this matrix); the resulting total number possible representations of $k$ nodes (which is simply $2^{b(k)}$); the number of canonical isomorphs $NC(k)$; and the number of canonical automorphism orbits. Note that, to map each possible set of $k$ nodes to their canonical isomorphs, the lookup table has $2^{b(k)}$ entries, and each entry has a value between 0 and $NC(k)-1$. Note that for $k$ up to 8, the graphettes can be stored in 32 bits. In that case, the maximum space required will be $32\times 2^{28}= 1$ GB. This is as far as we go, for now.  Moore's Law suggests that we may be able to go to $k=9$ within a few years, and to $k=10$ in perhaps a decade or two.

We note that the most expensive part of our algorithm is creating the lookup table between an arbitrary set of $k$ nodes, to the canonical graphette represented by those $k$ nodes; in the absence of a requirement for this lookup table, one could use orbit counting equations  \cite{Melckenbeeck1EtAl2016} to generate automorphism orbits up to $k=12$.

\subsection*{Generating the lookup table from non-canonical to canonical graphettes}
Assume the large graph $G$ has $n$ nodes labeled 0 through $n-1$, and pick an arbitrary set of $k$ nodes $U=\{u_0,u_1,\ldots,u_{k-1}\}$. Create the subgraph $g$ induced on the nodes in $U\subseteq \mathcal V(G)$, and let its bit vector representation $B$ be of the form lower-triangular matrix described in Fig \ref{fig:lookup-table}. We now describe how to create the lookup table that maps any such $B$ to its canonical representative.

We iterate through all $2^{b(k)}$ bit vectors in order; for each value $B$, we check to see if it is isomorphic to any of the previously found canonical graphettes; if so, the lookup table value is set to the previously found canonical graphette; otherwise we have a new, previously unseen canonical graphette and the lookup table value is set to itself ($B$).

When checking for isomorphism between $B$ and all previously found canonical graphettes, we use a relatively simple brute force approach. If the degree distribution of the two graphettes are different, we can immediately discard the pair as non-isomorphic; otherwise we resort to cycling through every permutation of the nodes checking each pair for graph equality, which has worst-case running time of $k^2k!$.
The total run time to compute the lookup table for a particular value $k$ is thus bounded above by $k^2k!\cdot NC(k)\cdot 2^{b(k)}$, where $k!$ is the maximum number of permutations we need to check if a non-canonical matches an existing canonical, $k^2$ is the worst-case running time to check if 2 specific permutations of $k$-graphettes are isomorphic, there are at most $NC(k)$ canonicals to check against \cite{numcanon}, and $2^{b(k)}=2^{n(n-1)/2}$ is the total number of undirected graphs on $k$ nodes. More sophisticated approaches exist  \cite{NAUTY}, which may more easily allow higher values of $k$.

This process can also be parallelized, which is what we did for $k=8$. Essentially, we can split the $2^{b(k)}$ non-canonical graphettes into $m$ sets of about $2^{b(k)}/m$ graphettes each, and then spread the computation across $m$ machines.  For each of the $m$ sets $S_i$, we loop through all graphettes in that set and mark out which are isomorphic to each other.  For each set $S_i$, we will find a set $T_i$ of lowest-numbered ``temporary'' canonical graphettes in $S_i$, along with the map $TC: S_i \rightarrow T_i$ of which graphettes in $S_i$ map to each temporary canonical in $T_i$.  That is, for each graphette $g\in S_i$, $\exists h\in T_i$ for which the temporary canonical $TC(g)=h$.  Finally, once all the $m$ sets have been evaluated in this way, a second stage passes through all the $T_i, i=0,\ldots,m-1$, merging the temporary canonicals together into a final, global list of canonical graphettes, while also propagating these globally lowest-numbered canonicals back up through the $m$ temporary canonical maps, so each graphette $g$ globally maps to the globally lowest-numbered canonical; we call this process {\it sifting for canonicals}, and it may require several iterations to globally find the final list of canonicals.  In this way we ran $k=8$ in about a week across 600 cores, for a total of 600 CPU-weeks.  This process could probably be made more efficient with smarter isomorphism checking \cite{NAUTY,SAUCY3}.

\subsection*{Graph automorphism and orbits}
An isomorphism $\pi:\mathcal V(g)\to\mathcal V(g)$ (from a graph $g$ to itself) is called an \textit{automorphism}.

While an isomorphism is just a permutation of the nodes, it is called an \textit{automorphism} if it results in exactly the same labeling of the nodes in the same order---in other words exactly the same adjacency matrix. The set of all automorphisms of $g$ will be called $Aut(g)$.

An \textit{automorphism orbit}, or just {\it orbit}, of $g$ is a minimally sized collection of nodes from $\mathcal V(g)$ that remain invariant under \textbf{every} automorphism of $g$  \cite{automorph-orbit}. There can be more than one automorphism orbit, and each orbit can have anywhere from $1$ to $k$ member nodes; refer again to Fig \ref{fig:graphlets5} for some examples. More formally, a set of nodes $\omega$ constitute an orbit of $g$ iff:
\begin{enumerate}
	\item For any node $u\in \omega$ and \textbf{any} automorphism $\pi$ of $g$, $u\in \omega \Longleftrightarrow \pi(u)\in \omega$.
	\item if nodes $u,v\in \omega$, then there exists an automorphism $\pi$ of $g$ and a $\gamma>0$ so that $\pi^{\gamma}(u)=v$.
\end{enumerate}
Now, we shall prove a few relevant results that will be useful later for automatically enumerating the orbits.
\begin{proposition}
	For each node $u\in \mathcal V(g)$ and each automorphism $\pi:\mathcal V(g)\to\mathcal V(g)$, there exists an integer $\lambda>0$ such that $\pi^\lambda(u)=u$.
\end{proposition}
\begin{proof}
	Because $\pi$ is an automorphism,
	\begin{align*}
		u\in \mathcal V(g) \implies& \pi(u)\in \mathcal V(g)\\
		                   \implies& \pi^2(u)\in \mathcal V(g)\\
		                   \vdots&\\
		                   \implies& \pi^i(u)\in\mathcal V(g), \;\; \forall i\in \mathbf{N}.
	\end{align*}
	Since $|\mathcal V(g)|$ is finite and $\pi$ is bijective, the conclusion obviously follows.
\end{proof}

We shall call the set of nodes \[\mathcal C_{\pi}(u) = \{u, \pi(u),\ldots, \pi^{\lambda-1}(u)\}\] the \textit{cycle} of $u$ under automorphism $\pi$, where $\lambda$ is the smallest positive integer such that $\pi^\lambda(u) = u$. 

Note that $\lambda$ is not unique since $\pi^\lambda(u) = \pi^{2\lambda}(u)=\cdots=u$. Also, $\pi, u$, and $\lambda$ are tied together into triples such that knowing any two determines the third.

\begin{corollary}\label{corollary:cycle1}
$\pi$ maps every node $\in \mathcal C_\pi(u)$ to a node (possibly same) $\in \mathcal C_\pi(u)$.
\end{corollary}
\begin{corollary}
In any automorphism $\pi$ of $g$, every node appears in exactly one cycle.
\end{corollary}
In other words, the cycles $\pi$ creates are disjoint. (However, the cycles from different automorphisms might not be so.) Hence, it makes sense to say \textit{splitting an automorphism into its cycles}. For example consider the permutation $\pi = (201354)$ of (012345). Since $\pi(0)=2, \pi(2)=1, \pi(1)=0$, the nodes (012) form a cycle. Now start with the next node, 3. $\pi(3) = 3$. So, (3) is another cycle. Finally, $\pi(4) = 5, \pi(5)= 4$, so, (45) form another cycle. Hence, the permutation (201354) is split into three cycles, namely (012), (3), (45).

\begin{proposition}
    The orbits are disjoint. (In other words, each node appears in exactly one orbit.)
\end{proposition}
\begin{proof}
    Assume the contrary, i.e., a node $u\in \mathcal V(g)$ appears in two different orbits $\omega_1$ and $\omega_2$. According to the second condition, for any other node $v\in \omega_1$, there exists an automorphism $\pi$ of $g$ and a $\gamma$ so that $\pi^\gamma(u)=v$. However, from the first condition,
    \begin{align*}
    u\in \omega_2 \implies& \pi(u)\in\omega_2\\
    \implies& \pi^2(u)\in\omega_2\\
    \vdots&\\
    \implies& \pi^\gamma(u)\in\omega_2\\
    \implies& v\in \omega_2\\
    \end{align*}
    Therefore, every node $v\in\omega_1$ also belongs to $\omega_2$. Hence, $\omega_1\subseteq\omega_2$.
    
    Following the same logic, $\omega_2\subseteq\omega_1$, implying $\omega_1=\omega_2$.  $\Rightarrow\!\Leftarrow$
\end{proof}
\begin{corollary}
Each cycle appears in exactly one orbit, which completely contains that cycle.
\end{corollary}
\begin{proof}
	If an orbit $\omega$ partially contains a cycle $\mathcal C_\pi(u)$, then $\omega$ is not invariant under automorphism $\pi$, as $\pi$ will map some node in $\omega$ (and $\mathcal C_{\pi}(u)$) to another node outside $\omega$ (but still in $\mathcal C_{\pi}(u)$) according to corollary \ref{corollary:cycle1}, contradicting our definition of orbits. Since two orbits are disjoint, $\mathcal C_\pi(u)$ must appear only in $\omega$, and in none of the other orbits. 
\end{proof}
These statements are enough to be able to find all orbits of each graphette, as we now demonstrate.

\subsection*{Automatically enumerating all orbits of a graph}
From the propositions in the previous section, an algorithm to enumerate the orbits can be constructed like this:
\begin{enumerate}
    \item Generate all automorphisms of $g$.
    \item Split each automorphism into its cycles.
    \item Merge the cycles from different automorphisms to form orbits.
\end{enumerate}
\begin{algorithm*}
	\caption{Automatically enumerating automorphism orbits of a graph}
	\label{generateAutomorphs}
	\begin{algorithmic}
        \tt
        \Function{\sc generateAutomorphisms}{Graph $g$}
		    \State $Aut(g)=\{\}$ \; // Find the automorphisms of $g$
		    \For {each permutation $\pi$ of $\mathcal V(g)$}
		        \State apply $\pi$ over $Adj(g)$
		        \If{$Adj(g)==\pi(Adj(g))$} \textbf{put} $\pi$ in $Aut(g)$
		        \EndIf
		    \EndFor
        \EndFunction
		\\
		\Function{\sc generateCycles}{automorphism $\pi$}
    		\State $\mathcal C=\{\}$
    		\For{node $u$ in $\pi$}
    		    \If{$u$ is \textbf{not} visited}
    		        \State \textbf{mark} $u$ visited
    		        \State new cycle $\mathcal C_\pi(u) = \{\}$
    		        \State node $v$ = $\pi(u)$
    		        \While{$v$ != $u$}
    		            \State \textbf{put} $v$ in $\mathcal C_\pi(u)$
    		            \State \textbf{mark} $v$ visited
    		            \State $v = \pi(v)$
    		        \EndWhile
    		    \State \textbf{put} $\mathcal C_\pi(u)$ in $\mathcal C$
    		    \EndIf 
    		\EndFor
		\EndFunction
		\\
		\Function{\sc enumerateOrbits}{$\mathcal C(g)$}
		\For {each node $u\in\mathcal V(g)$} $\omega(u)=u$
		\EndFor
		\For{cycle $c\in \mathcal C(g)$}
			\State \textbf{let} $\omega_{\min}=\infty$
			\For{node $u\in c$} $\omega_{\min} = \min(\omega_{\min}, \omega(u))$
			\EndFor
			\For{node $u\in c$} $\omega(u) = \omega_{\min}$
			\EndFor
		\EndFor 
		\EndFunction
	\end{algorithmic}
\end{algorithm*}

\subsubsection*{Generating all automorphisms of $g$}
Referring to Algorithm 1, the function \textsc{generateAutomorphisms()} applies every possible permutation of $\mathcal V(g)$ over $Adj(g)$. Each permutation creates an isomorph of $Adj(g)$. If $Adj(g)$ is unchanged under some permutation $\pi$, then by definition, $\pi$ is an automorphism of $g$. Hence it is saved into $Aut(g)$. 

Two optimization strategies are employed:
\begin{enumerate}
	\item No node is mapped to another node with unequal degree.
	\item An automorphism of graph $g$ is also an automorphism of its complement graph $g'$. 
\end{enumerate}

In practice, this algorithm generates all automorphisms of all the canonical graphettes up to size 8 in a matter of seconds. Nevertheless, for additional speed up in higher sizes, modern sophisticated automorphism detection algorithms \cite{NAUTY, SAUCY3} may be used. 

\subsubsection*{Splitting automorphisms into cycles}
An automorphism $\pi$ of $g$ is basically a permutation of nodes of $g$. Hence, to split $\pi$ into cycles, we can repeatedly apply $\pi$ over every node $u\in\pi$ and remember the nodes $u$ transforms into. This forms the cycle with node $u$, i.e. $\mathcal C_\pi(u)$, which is saved in $\mathcal{C}$. After first visit, each node is marked visited to prevent more visits. 

\subsubsection*{Merging cycles to enumerate orbits}
Suppose $\mathcal C(g)$ is the set of all cycles resulting from all the automorphisms of $g$. 

To enumerate orbits from it, first each node $u$ is colored with a unique color $\omega(u)=u$. Then $\omega(u)$ is continuously updated to reflect the current color of $u$, as the nodes belonging to same orbits are gradually colored by identical color.

For the nodes of each cycle $c\in\mathcal C(g)$, we save their minimum color in $\omega_{\min}$, and then color all of them with $\omega_{\min}$. After coloring all the cycles in this way, nodes belonging to same orbits get the same color, and hence, get enumerated.

\subsection*{Proof of correctness of Algorithm  \ref{generateAutomorphs}}
Here we prove that Algorithm \ref{generateAutomorphs} determines every orbit of $g$.

Suppose a set $\omega$ is among the final sets generated by Algorithm \ref{generateAutomorphs}. We shall prove $\omega$ is an orbit of $g$ by showing that it follows the two properties of orbits:
\begin{enumerate}
	\item Let a node $u\in\omega$ form the cycle $\mathcal C_{\pi}(u)$ under automorphism $\pi$. The \textsc{generateCycles} function will apply $\pi$ repeatedly until it finds a $\lambda$ so that $\pi^\lambda(u)=u$ and will therefore determine $\mathcal{C}_{\pi}(u)$. Since the \textsc{enumerateOrbits} function assigned $u$ to $\omega$, it had also assigned all nodes in $\mathcal C_{\pi}(u)$ to $\omega$. Hence $u\in\omega\Longleftrightarrow \pi(u)\in\omega$.
	
	\item Suppose nodes $u,v\in\omega$. Then, either they belonged to a cycle from which they were assigned to a mutual set $\omega$ in \textsc{enumerateOrbits} function, or there is a third node $w$ so that $w$ shares separate cycles with $u$ and $v$ under different automorphisms $\pi_1$ and $\pi_2$. In the first case, $u$ and $v$ already belong to a common cycle. In the second case, assume $\pi_1^{\gamma_1}(w)=u$ and $\pi_2^{\gamma_2}(w)=v$. Consider the permutation $\phi=\pi_2^{\gamma_2}\circ\pi_1^{-\gamma_1}$. Since composition of two automorphisms is an automorphism \cite{automorph-compose}, $\phi$ is also an automorphism. And notice that \[\phi(u) = \pi_2^{\gamma_2}\left(\pi_1^{-\gamma_1}(u)\right)= \pi_2^{\gamma_2}(w)=v\]
	implying $u$ and $v$ belong to a common cycle under $\phi$.
\end{enumerate}
Therefore, $\omega$ is indeed an orbit of $g$. Since each node was given a unique orbit color in the beginning of {\sc enumerateOrbits}, every orbit of $g$ will be eventually found by Algorithm \ref{generateAutomorphs}.

% Results and Discussion can be combined.
\section*{Results and discussion}
Using the algorithms described herein, we have enumerated all possible graphlets, including the generalization of disconnected counterparts called {\it graphettes}, up to size $k=8$. The code and data can be found in \url{http://github.com/Neehan/Faye}. (Note that the github code uses the {\em upper} triangle matrix, though we intend to convert it to use the lower tringle as that representation has already been established \cite{Melckenbeeck1EtAl2016}.) We have also enumerated all orbits up to size $k=8$. More importantly to the statistical sampling technique described in the Introduction, we have used a bit-vector representation of all possible adjacency matrices of all possible sets of up to $k=8$ nodes and created a lookup table from the $2^{k(k-1)/2}$ $k$-sets to their canonical graphette representatives.  This allows us to determine, in constant time, the graphette represented by these $k$ nodes, as well as the automorphism orbits of each nodes. This allows efficient estimation of both the global distribution of graphlets and orbits, as well as an estimation of the graphlet (or orbit) degree vector for each node in a large graph $G$.

Although the lookup tables for $k>8$ are at present too big to compute or store, we could also use NAUTY or SAUCY to enumerate all the canonical graphettes up to size $k=12$, and use our orbit generation code Algorithm 1 to determine all the orbits in all graphettes up to size $k=12$.
We have verified that previous results are consistent with ours in terms of the number of distinct graphettes  \cite{numcanon} and orbits  \cite{numorbit} determined, as displayed in Table \ref{tab:k-graphette-canonical-orbit}.

In future work we will study which statistical sampling techniques most efficiently produce a good estimate of the complete graphlet and local (per-node) degree vectors. We also intend to study how this method may aid in cataloging of graphlets for database network queries, or in non-alignment network comparison  \cite{Przulj2014HiddenLanguage}.
Finally, there may be ways to combine our method with those of orbit counting equations  \cite{ORCA,Melckenbeeck1EtAl2016} to more efficiently produce samples of orbit counts.

% Place tables after the first paragraph in which they are cited.
%\begin{table}[!ht]
%\begin{adjustwidth}{-2.25in}{0in} % Comment out/remove adjustwidth environment if table fits in text column.
%\centering
%\caption{
%{\bf Table caption Nulla mi mi, venenatis sed ipsum varius, volutpat euismod diam.}}
%\begin{tabular}{|l+l|l|l|l|l|l|l|}
%\hline
%\multicolumn{4}{|l|}{\bf Heading1} & \multicolumn{4}{|l|}{\bf Heading2}\\ \thickhline
%$cell1 row1$ & cell2 row 1 & cell3 row 1 & cell4 row 1 & cell5 row 1 & cell6 row 1 & cell7 row 1 & cell8 row 1\\ \hline
%&$cell1 row2$ & cell2 row 2 & cell3 row 2 & cell4 row 2 & cell5 row 2 & cell6 row 2 & cell7 row 2 & cell8 row 2\\ \hline
%$cell1 row3$ & cell2 row 3 & cell3 row 3 & cell4 row 3 & cell5 row 3 & cell6 row 3 & cell7 row 3 & cell8 row 3\\ \hline
%\end{tabular}
%\begin{flushleft} Table notes Phasellus venenatis, tortor nec vestibulum mattis, massa tortor interdum felis, nec pellentesque metus tortor nec nisl. Ut ornare mauris tellus, vel dapibus arcu suscipit sed.
%\end{flushleft}
%\label{table1}
%\end{adjustwidth}
%\end{table}

%PLOS does not support heading levels beyond the 3rd (no 4th level headings).

\section*{Acknowledgments}
We thank Sridevi Maharaj, Dillon Kanne, and the anonymous referees for several helpful suggestions on presentation.

\nolinenumbers

% Either type in your references using

%
% or
%
% Compile your BiBTeX database using our plos2015.bst
% style file and paste the contents of your .bbl file
% here. See http://journals.plos.org/plosone/s/latex for 
%\bibliography{wayne}

\begin{thebibliography}{10}

\bibitem{Cook:1971:CTP:800157.805047}
Cook SA.
\newblock The Complexity of Theorem-proving Procedures.
\newblock In: Proceedings of the Third Annual ACM Symposium on Theory of
  Computing. STOC '71. New York, NY, USA: ACM; 1971. p. 151--158.
\newblock Available from: \url{http://doi.acm.org/10.1145/800157.805047}.

\bibitem{NewmanNetworks2010}
Newman M.
\newblock Networks: an introduction. 2010.
\newblock United Slates: Oxford University Press Inc, New York. 2010; p. 1--2.

\bibitem{emmert2016fifty}
Emmert-Streib F, Dehmer M, Shi Y.
\newblock Fifty years of graph matching, network alignment and network
  comparison.
\newblock Information Sciences. 2016;346:180--197.

\bibitem{WilsonZhu2008Spectral}
Wilson RC, Zhu P.
\newblock A study of graph spectra for comparing graphs and trees.
\newblock Pattern Recognition. 2008;41(9):2833--2841.

\bibitem{ThorneStrumpf2012}
Thorne T, Stumpf MP.
\newblock Graph spectral analysis of protein interaction network evolution.
\newblock Journal of The Royal Society Interface. 2012; p. rsif20120220.

\bibitem{dehmer2014interrelations}
Dehmer M, Emmert-Streib F, Shi Y.
\newblock Interrelations of graph distance measures based on topological
  indices.
\newblock PloS one. 2014;9(4):e94985.

\bibitem{Milo2002Motifs}
Milo R, Shen-Orr S, Itzkovitz S, Kashtan N, Chklovskii D, Alon U.
\newblock Network motifs: simple building blocks of complex networks.
\newblock Science. 2002;298(5594):824--827.

\bibitem{Przulj2004Graphlets}
Pr{\v{z}}ulj N, Corneil DG, Jurisica I.
\newblock Modeling interactome: scale-free or geometric?
\newblock Bioinformatics. 2004;20(18):3508--3515.

\bibitem{PrzuljGDD2007}
Pr{\v{z}}ulj N.
\newblock Biological network comparison using graphlet degree distribution.
\newblock Bioinformatics. 2007;23(2):e177--e183.

\bibitem{Przulj2014HiddenLanguage}
Yavero{\u{g}}lu {\"O}N, Malod-Dognin N, Davis D, Levnajic Z, Janjic V,
  Karapandza R, et~al.
\newblock Revealing the hidden language of complex networks.
\newblock Scientific reports. 2014;4:4547.

\bibitem{GRAAL}
Kuchaiev O, Milenkovi{\'c} T, Memi{\v s}evi{\'c} V, Hayes W, Pr\v{z}ulj N.
\newblock Topological network alignment uncovers biological function and
  phylogeny.
\newblock Journal of The Royal Society Interface. 2010;7(50):1341--1354.
\newblock doi:{10.1098/rsif.2010.0063}.

\bibitem{LGRAAL}
Malod-Dognin N, Pr\v{z}ulj N.
\newblock L-GRAAL: Lagrangian Graphlet-based Network Aligner.
\newblock Bioinformatics. 2015;doi:{10.1093/bioinformatics/btv130}.

\bibitem{MAGNA}
Saraph V, Milenkovi{\'c} T.
\newblock MAGNA: maximizing accuracy in global network alignment.
\newblock Bioinformatics. 2014;30(20):2931--2940.

\bibitem{MamanoHayesSANA}
Mamano N, Hayes W.
\newblock SANA: Simulated Annealing far outperforms many other search
  algorithms for biological network alignment.
\newblock Bioinformatics. 2017;0(0):8.

\bibitem{ORCA}
Ho\v{c}evar T, Dem\v{s}ar J.
\newblock {A combinatorial approach to graphlet counting}.
\newblock Bioinformatics. 2014;30(4):559--565.
\newblock doi:{10.1093/bioinformatics/btt717}.

\bibitem{Melckenbeeck1EtAl2016}
Melckenbeeck I, Audenaert P, Michoel T, Colle D, Pickavet M.
\newblock An Algorithm to Automatically Generate the Combinatorial Orbit
  Counting Equations.
\newblock PLoS ONE. 2016;11(1).
\newblock doi:{http://dx.doi.org/10.1371/journal.pone.0147078}.

\bibitem{Chatr-aryamontri01012013}
Chatr-aryamontri A, Breitkreutz BJ, Heinicke S, Boucher L, Winter A, Stark C,
  et~al.
\newblock The BioGRID interaction database: 2013 update.
\newblock Nucleic Acids Research. 2013;41(D1):D816--D823.
\newblock doi:{10.1093/nar/gks1158}.

\bibitem{NDEx2017}
Pillich RT, Chen J, Rynkov V, Welker D, Pratt D.
\newblock NDEx: A Community Resource for Sharing and Publishing of Biological
  Networks.
\newblock Protein Bioinformatics: From Protein Modifications and Networks to
  Proteomics. 2017; p. 271--301.

\bibitem{blast}
Camacho C, Coulouris G, Avagyan V, Ma N, Papadopoulos JS, Bealer K, et~al.
\newblock BLAST+: architecture and applications.
\newblock BMC Bioinformatics. 2009;10:421.

\bibitem{rahman2014graft}
Rahman M, Bhuiyan MA, Al~Hasan M.
\newblock Graft: An efficient graphlet counting method for large graph
  analysis.
\newblock IEEE Transactions on Knowledge and Data Engineering.
  2014;26(10):2466--2478.

\bibitem{prvzulj2006efficient}
Pr{\v{z}}ulj N, Corneil DG, Jurisica I.
\newblock Efficient estimation of graphlet frequency distributions in
  protein--protein interaction networks.
\newblock Bioinformatics. 2006;22(8):974--980.

\bibitem{numcanon}
Sloane N. Online Encyclopedia of Integer Sequences (OEIS);.
\newblock Available from: \url{http://oeis.org/A000088}.

\bibitem{NAUTY}
Mckay BD. Nauty; 2010.
\newblock Available from: \url{http://users.cecs.anu.edu.au/~bdm/nauty}.

\bibitem{SAUCY3}
Codenotti P, Katebi H, Sakallah KA, Markov IL.
\newblock Conflict Analysis and Branching Heuristics in the Search for Graph
  Automorphisms.
\newblock In: Tools with Artificial Intelligence (ICTAI). IEEE; 2013.

\bibitem{automorph-orbit}
Gross JL. Graph Theory – Lecture 2: Structure and Representation — Part A;.
\newblock Available from:
  \url{http://www.cs.columbia.edu/~cs4203/files/GT-Lec2.pdf}.

\bibitem{automorph-compose}
Automorphism of a group;.
\newblock Available from:
  \url{https://groupprops.subwiki.org/wiki/Automorphism_of_a_group}.

\bibitem{numorbit}
Sloane N. Online Encyclopedia of Integer Sequences (OEIS);.
\newblock Available from: \url{http://oeis.org/A000666}.

\end{thebibliography}

\end{document}